\documentclass{article}

\usepackage[preprint]{neurips_2022}

\usepackage[utf8]{inputenc} 
\usepackage[T1]{fontenc}    
\usepackage{hyperref}       
\usepackage{url}            
\usepackage{booktabs}       
\usepackage{amsfonts}       
\usepackage{nicefrac}       
\usepackage{microtype}      
\usepackage{xcolor}         

\usepackage{amsmath}
\usepackage{amssymb}
\usepackage{mathtools}
\usepackage{amsthm}

\usepackage[ruled]{algorithm2e}

\usepackage{float}
\usepackage{graphicx}
\usepackage{caption}
\usepackage{subfig}
\usepackage{tikz}
\usetikzlibrary{arrows,positioning} 
\usepackage{tcolorbox}
\usetikzlibrary{decorations.pathreplacing,calc}
\usepackage[export]{adjustbox}

\usepackage[capitalize,noabbrev]{cleveref}
\newcommand{\At}{\mathcal{A}_t}
\newcommand{\St}{S_t}

\newcommand{\MBk}{MB_{k}}

\newcommand{\MBkhat}{\widehat{N}_{k}}

\newcommand{\MBkthat}{\widehat{N}_{kt}}
\newcommand{\MBjthat}{\widehat{N}_{jt}}
\newcommand{\Lkt}{L_{kt}}
\newcommand{\Ltj}{L_{tj}}
\newcommand{\Rkt}{R_{kt}}
\newcommand{\rkt}{r_{kt}}
\newcommand{\Rjt}{R_{jt}}

\newcommand{\Rlt}{R_{\ell t}}
\newcommand{\rlt}{r_{\ell t}}
\newcommand{\M}{\mathbf{M}}
\newcommand{\B}{\mathbf{B}}
\newcommand{\eps}{\epsilon}
\newcommand{\X}{\mathbf{X}}

\newcommand{\iid}{\stackrel{\rm i.i.d.}{\sim}}

\newcommand{\E}{\mathbb{E}}
\newcommand{\V}{\mathbb{V}}

\newcommand{\var}{{\rm Var}}

\newcommand{\indep}{\mbox{$\perp\!\!\!\perp$}}
\newcommand{\norm}[1]{\left\lVert#1\right\rVert}
\theoremstyle{plain}
\newtheorem{theorem}{Theorem}[section]

\newtheorem{lemma}[theorem]{Lemma}

\theoremstyle{definition}
\newtheorem{definition}[theorem]{Definition}
\newtheorem{assumption}[theorem]{Assumption}
\theoremstyle{remark}
\newtheorem{remark}[theorem]{Remark}

\usepackage[textsize=tiny]{todonotes}

\title{Sequentially learning the topological ordering of causal directed acyclic graphs with likelihood ratio scores}

\author{%
  Gabriel Ruiz\thanks{Gabriel would like to thank NSF-GRFP DGE-1650604 for financial support.} \\
  Department of Statistics\\
  University of California, Los Angeles\\
  Los Angeles, CA 90095\\
  \texttt{ruizg@ucla.edu} \\
   \And
  Oscar Hernan Madrid Padilla\\
  Department of Statistics\\
  University of California, Los Angeles\\
  Los Angeles, CA 90095\\
  \texttt{oscar.madrid@stat.ucla.edu} \\
   \AND
  Qing Zhou\\
  Department of Statistics\\
  University of California, Los Angeles\\
  Los Angeles, CA 90095\\
  \texttt{zhou@stat.ucla.edu} \\
}

\begin{document}

\maketitle

\begin{abstract}
Causal discovery, the learning of causality in a data mining scenario, has been of strong scientific and theoretical interest as a starting point to identify ``what causes what?'' Contingent on assumptions and a proper learning algorithm, it is sometimes possible to identify and accurately estimate a causal directed acyclic graph (DAG), as opposed to a Markov equivalence class of graphs that gives ambiguity of causal directions. The focus of this paper is in highlighting the identifiability and estimation of DAGs with general error distributions through a general sequential sorting procedure that orders variables one at a time, starting at root nodes, followed by children of the root nodes, and so on until completion. We demonstrate a novel application of this general approach to estimate the topological ordering of a DAG. At each step of the procedure, only simple likelihood ratio scores are calculated on regression residuals to decide the next node to append to the current partial ordering. The computational complexity of our algorithm on a $p$-node problem is $\mathcal{O}(pd)$, where $d$ is the maximum neighborhood size. Under mild assumptions, the population version of our procedure provably identifies a true ordering of the underlying DAG. We provide extensive numerical evidence to demonstrate that this sequential procedure scales to possibly thousands of nodes and works well for high-dimensional data. We accompany these numerical experiments with an application to a single-cell gene expression dataset.
\end{abstract}

\section{Introduction}

With observational data alone, causal inference using an accurate directed acyclic graph (DAG) has been shown to provide results that are up to par with the quintessential randomized controlled experiment 
\citep{pearlCausality,pmlr-v89-malinsky19b_poCalc}. However, it is difficult to imagine that this approach, with its reliance on strong domain knowledge about the system of variables at hand, can be applied to cases with large numbers of variables and little background on how they are all related, such as in bioinformatics and fields of science where ``big data” was previously unavailable and we are now trying to get a grasp of it. On the other hand, causal discovery—learning the DAG structure for Bayesian networks from scratch—has its own pitfalls, such as a super-exponentially increasing number of networks in the search space as the number of nodes grows and the fact that it is possible for multiple directed acyclic graphs and the Bayesian networks they encode to map to the same joint distribution--a phenomenon called Markov equivalence \citep{Peters2017}. Nonetheless, a large  effort has been devoted to the structure learning of DAGs as a preliminary data mining step in the scientific pipeline.

The present paper provides discussion of a general learning algorithm, which we show can be quite scaleable, along with novel identification theory for a specific application. The main task of the approach we advocate for is to sequentially estimate a topological ordering of the DAG, a permutation of node labels such that every parent must precede its children. To help with scalability in practice, we also make use of a priori known neighborhood sets, such as a Markov blanket of a node. In order to demonstrate the theoretical promise of this procedure, we discuss existing identification results that make use of it. We also provide new theory for a linear structural equation model (SEM) first studied in \citet{JMLR:v7:shimizu06a}. The novelty of our application of the sequential sorting procedure to this SEM compared to the state-of-the-art for it is the scalability of our procedure to a large number of nodes in the underlying graphical model.


Representative methods for causal discovery under the assumption of no unobserved confounding (causal sufficiency) 
include the Peter-Clark (PC) algorithm \citep{pcAlgo1991} and Greedy Equivalence Search (GES) \citep{chickering2002}. The PC algorithm is a constraint-based method due to its use of conditional independence queries, while GES is considered a score-based method for the objective function it seeks to optimize across the space of graphical models. Without additional structural assumptions, the best these methods can generally do in the limit of sample size ($n\to\infty$) is to obtain a Markov equivalence class (MEC) of DAGs, visualized typically by a single Completed Paritially Directed Acyclic Graph (CPDAG) as in Figure \ref{cpdagExample}. Each DAG in the MEC, obtainable by orienting undirected edges in the CPDAG without introducing a cyclic path nor a ``v-structure,'' encodes the same set of \textit{d-separation} relations that imply marginal and conditional independence relations between triplets of variable subsets in their underlying joint distribution \citep{pcAlgo1991}. 

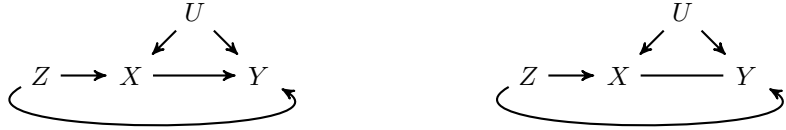
\begin{figure}%
\centering
\begin{tikzpicture}[-,>=stealth',shorten >=1pt,auto,node distance=1.2cm,
thick,main node/.style={circle,fill=blue!20,draw,font=\sffamily\Large\bfseries}]
\node (A) {$U$};
\node (X) [below left of=A] {$X$};
\node (Z) [left of=X] {$Z$};
\node (Y) [below right of=A] {$Y$};
\draw[->] (A) -- (X) ;
\draw[->] (A) --  (Y);
\draw[->] (X) --  (Y);
\draw[->] (Z) -- (X);
\draw[->] (Z) to [out=-150,in=-30] (Y);
\end{tikzpicture}\qquad
\begin{tikzpicture}[-,>=stealth',shorten >=1pt,auto,node distance=1.2cm,
  thick,main node/.style={circle,fill=blue!20,draw,font=\sffamily\Large\bfseries}]
  \node (A) {$U$};
  \node (X) [below left of=A] {$X$};
  \node (Z) [left of=X] {$Z$};
  \node (Y) [below right of=A] {$Y$};
  \draw[->] (A) -- (X) ;
  \draw[->] (A) --  (Y);
  \draw[-] (X) --  (Y);
  \draw[->] (Z) -- (X);
  \draw[->] (Z) to [out=-150,in=-30] (Y);
\end{tikzpicture}
\caption{The original DAG (left) and its corresponding CPDAG (right), obtained by keeping the orientation of edges corresponding to the v-structures $Z\to X\leftarrow U$ and $Z\to Y\leftarrow U$, and removing the orientation from all other edges. Note the ambiguity about the causal direction $X\to Y$ vs. $X\leftarrow Y$ in the CPDAG.\label{cpdagExample}}
\end{figure}

When additional assumptions are justified, such as strict non-linearity of structural equations, or non-Gaussianity of noise terms in a linear structural equation model, a unique DAG can be identified \citep{buhlmann2014,JMLR:v7:shimizu06a}. When we are not willing to make the assumption of causal sufficiency, the Fast Causal Inference (FCI) algorithm provides an alternative at the cost of a potentially less precise, though possibly more accurate, graphical model compared to a DAG or CPDAG \citep{spirtes2000constructing}. Beyond what we highlight here for the case of iid samples from a distribution that our DAG of interest satisfies the Markov property with respect to, \citet{reviewShort} and \citet{Peters2017} provide reviews on the trade-offs of different algorithms and what can and cannot be done when there is additional structure, such as the case that the system of variables varies in time. In the context of Earth system sciences, \citet{rungeClimate} review causal discovery methods. Structure learning has also been explored for its possibility to explain the black-box nature of state-of-the-art deep learning architectures \citep{sani_blackbox}. 
Moreover, \citet{ZhengNoTears2018_NEURIPS_e347c514} and its extension to \citet{zheng20aSparseNonParamDAGs} provide an approach to optimize a non-convex score function in DAG space by using a smooth characterization of an adjacency matrix's acyclicity constraint. 

\subsection{Review of relevant work} 

Specific to our task of learning a topological ordering for an underlying acyclic graphical model, we now review some relevant work. Let us first formally define a so-called topological ordering of a DAG, the target parameter we seek to estimate. Let $[m]=\{1,\dots,m\}$ for integer $m\geq 1$. 
\begin{definition}\label{defn:perm} 
A topological ordering for a DAG $\mathcal{G}$ is given by a permutation 
$\pi:[p]\to[p]$
such that every parent node precedes its child in the ordering: 
\[
  j\in PA_k\implies \pi^{-1}(j)<\pi^{-1}(k). 
\]
\end{definition}

Importantly, we note that the discrete search space across $p!$ permutation functions in search of one that satisfies Definition \ref{defn:perm} can be quite cumbersome \citep{raskuttiUhler2018SparsePerm,solus2021consistency}. Several heuristic score-based methods have been developed to cope with the search space, however, it remains the case that score-based approaches for ordering search are NP hard in general \citep{Chickering1996,yeSimAnnealing}. Along these lines, recent work by \citep{yeSimAnnealing} provides one approach under the case of a linear Bayesian network with Gaussian noise. The similar non-parametric approach of \citet{solus2021consistency} and \citet{wangSolusYangUhler2017} requires a consistent conditional independence testing procedure to decide the presence or absence of an edge in the DAG corresponding to a given $\tilde{\pi}$ in the search space. The empirical results of these approaches are all promising. However, these methods do not contain an application to more than $100$ nodes. Moreover, these works do not provide guarantees on whether the search can terminate at a point well before querying all permutations or DAGs in order to achieve optimal (statistically consistent) results.

Complementary to these advances, we here study a simple approach for which the search has a pre-determined number of steps: $\mathcal{O}(pd)$ in the number of least squares residual updates, where $d\leq p$ is the maximum neighborhood size of a node. The objective of our work is to estimate one such permutation $\pi$ from observed data $\X\in\mathbb{R}^{n\times p}$ of sample size $n$. 
We denote this estimate as $\hat{\pi}$. To do so, we will apply Algorithm~\ref{alg:continueOrdering} for $t=1,2,\dots,p$ until all nodes are sorted. At step $t$, given an input partial ordering $\At=(\hat{\pi}(1),\ldots,\hat{\pi}(t-1))$, the algorithm identifies a node $\hat{\pi}(t)\notin\At$ to append to the estimated ordering by maximizing the score $\mathcal{S}(k,\At;\X)$.

\begin{algorithm}\ \\
\caption{Continue a Topological Ordering}\label{alg:continueOrdering}
\KwData{The partial ordering $\At$ and data matrix $\X\in\mathbb{R}^{n\times p}$}
\KwResult{The continued partial ordering $\mathcal{A}_{t+1}$}
\For{$k\not\in\At$}{
$s_k \gets \mathcal{S}(k,\At;\X)$\\
}
$\hat{\pi}(t)\gets\arg\max_{k\not\in\At} s_k$\\
$\mathcal{A}_{t+1}\gets \At\cup\{\hat{\pi}(t)\}$
\end{algorithm}

There exist structure learning methods that use the general approach in Algorithm \ref{alg:continueOrdering} to sequentially construct a topological ordering. These approaches motivate our present work and include the following. \citet{JMLR:v15:peters14a} apply Algorithm \ref{alg:continueOrdering} under an assumption of strictly nonlinear structural equations with additive noise. Meanwhile \citet{ghoshalPolynomialTime}, \citet{park_kim_2020}, and \citet{chenDrtonWang} apply this sequential sorting procedure under a bounded conditional variance assumption: $a\leq \V[X_j|X_{PA_j}]\leq b$ for each $j\in[p]$ and some unknown positive constants $a\leq b$ restricted by the signal a parent sends its child node. \citet{park_kim_2020} can be considered the most general of the three similar approaches as it contains an extended discussion on the case of a node's possibly non-linear relation with its parents. \citet{nonParamRestrCondVar_NEURIPS2020_85c9f9ef} further explore the scaleability for the sequential application of Algorithm \ref{alg:continueOrdering} to estimate non-linear structural equation models under this bounded conditional variance assumption. With respect to linear SEMs, applications of Algorithm \ref{alg:continueOrdering} include \citet{shimizu11}, \citet{pairwiseLikelihoodRatios}, and \citet{wangHighDLingam}, while \citet{ZENG2020130} construct the topological ordering in reverse starting with child-less nodes. We believe there exists potential to scale up the estimation of each of these models. We focus on the linear SEM here.


\subsection{Paper Contribution and Outline}

Our application of Algorithm \ref{alg:continueOrdering} to a linear SEM with non-Gaussian noise under causal sufficiency, a model known as the linear non-Gaussian acyclic model (LiNGAM). In terms of theoretical guarantees for the estimation of LiNGAM, \citet{JMLR:v7:shimizu06a} and \citet{shimizu11} provide identifiabiliy results for the respective LiNGAM learning procedures--that is, with  knowledge of the true distribution defined by the LiNGAM and an oracle for conditional independence queries in the case of latter. Meanwhile \citet{wangHighDLingam} provide formal statistical consistency results for their LiNGAM-learning procedure. 

{Although the above methods have very nice theoretical guarantees, their practical application is limited as they do not presently scale well to large graphs, say with thousands of nodes, as confirmed in the numerical results in this paper. Therefore, we develop a fast sequential learning method that can estimate large graphs in practice. At each step of this method, a node is selected to append to a partial ordering, so that after $p$ steps, where $p$ is the number of nodes in the underlying graph, a full ordering of all the nodes will be produced.  Compared to the existing works on LiNGAM, the main contributions of our work are:
\begin{enumerate}
    \item Based on a specified error distribution, we define a novel likelihood ratio score which is used at each step in our sequential algorithm. The evaluation of the likelihood ratio only involves linear regression and residual calculation. There are no tuning parameters.
    \item We prove that at the population-level, this sequential algorithm will identify a true ordering of the underlying DAG under proper assumptions on the LiNGAM.
    \item Our sequential method is computationally tractable with computational complexity $O(p^2)$ for the number of updates used in the entire algorithm. If prior knowledge on the Markov blankets of the nodes is provided, the computational complexity can be further reduced to $O(pd)$, 
    where $d$ is the maximum size of the Markov blankets. 
    This is in sharp contrast to traditional score-based approaches for ordering search, which are NP hard in general \citep{Chickering1996,yeSimAnnealing}.
\end{enumerate}
}

The rest of the paper is organized as follows. In the rest of this section, we formally introduce the linear SEM of interest. Next, in Section \ref{methodology} we will introduce our approach: \S\ref{conditions} discusses the conditions for this approach to work; \S\ref{identifiability} provides a formal identifiability result; and \S\ref{FiniteSample} provides the finite sample version of the algorithm. Section \ref{simulations} presents simulation results for our procedure for small and large-sized Bayesian networks, along with an application to single-cell gene expression data. Finally, we conclude with a summary of our findings and discussion of future work. 





\subsection{{Review of LiNGAM}}
We follow closely here the definition of a LiNGAM given by \citet{JMLR:v7:shimizu06a}. 
\begin{definition}{(Linear Non-Gaussian Acyclic Model)}\label{defn:lingam}\\
For $p\geq 2$, let $\mathcal{G}$ be a DAG on $p$ nodes and $\mathbf{B}\in\mathbb{R}^{p\times p}$ 
be the weighted adjacency matrix of $\mathcal{G}$ such that $\mathbf{B}_{jk}\neq 0$ means $j\in PA_k$, the parent set of node $k$. Let $\epsilon=(\epsilon_1,\ldots,\epsilon_p)$ such that $\epsilon_k\sim g(\cdot;\theta_k)$ independently with $g(\cdot;\theta)$ a density of a non-Gaussian distribution parameterized by $\theta\in\mathbb{R}^q$. We say $X\in\mathbb{R}^p$ follows a LiNGAM with DAG $\mathcal{G}$ if
\begin{align}\label{eq:linSEM}
X_{k}=\sum_{j\in PA_k }\mathbf{B}_{jk}X_j +\epsilon_k, \quad\quad k=1,\ldots,p.
\end{align}
\end{definition}

The scalar form of the linear SEM in Equation \eqref{eq:linSEM} can be rewritten in vector form as $X= \mathbf{B}^TX+\epsilon.$ Put $\mathbf{M}=\left(\mathbb{I}_p-\mathbf{B}\right)^{-T}$, a matrix with ones on its diagonal. Let $AN_k$ denote the ancestor set of node $k$: $a\in AN_k$ means there exists a direct path starting at node $a$ and ending at node $k$, $a\to \dots \to k$. Then we arrive at $X= \mathbf{M}\epsilon$ in vector form and $X_k=\sum_{j\in AN_k\bigcup\{k\} }\mathbf{M}_{kj}\epsilon_j$ in scalar form for all $k\in[p]$. Noting that $\mathbf{M}$ serves as a mixing matrix for the independent components in $\eps$, we may think of the estimation of this linear SEM as an instance of Independent Component Analysis (ICA) \citep{HYVARINEN2000411}. \citet{JMLR:v7:shimizu06a} discuss the connection between LiNGAM and ICA.

\section{Methodology and algorithm\label{methodology}}

In this section, we introduce both the population-level and finite-sample versions of our sorting procedure. 
We also show that our choice of summary score $\mathcal{S}(k,\At;\X)$ 
in Algorithm~\ref{alg:continueOrdering} will lead to the identification of a topological ordering 
of the true DAG $\mathcal{G}$ used to define the linear SEM of Definition \ref{defn:lingam}. We start with a few main assumptions on the linear SEM we will work with.

\subsection{Assumptions\label{conditions}}

Our main assumptions are on the distributions of the independent errors $\epsilon$. We consider restricting our class of densities $\{g(\cdot;\theta_k)\}_{1\leq k\leq p}$ for the noise terms in Definition \ref{defn:lingam} to a scale-location family in which the $\theta_k>0$ are the scale parameters, such as the Laplace family of distributions, the Logistic family of distributions, or a Scaled-t distribution family (same degrees of freedom). This is summarized in Assumption \ref{scaleLocation}.

\begin{assumption}\label{scaleLocation}\ \\
Let $U\sim g(\cdot;\theta_0)$ with $\theta_0>0$ and $\E[U]=0$. For each $k=1,2,\dots,p$, the density of the error $\eps_k$ satisfies
\[
g(e;\theta_k)=\frac{\theta_0}{\theta_k}g(\theta_0 e/\theta_k;\theta_0).
\]
That is, $\eps_k\overset{d}{=}(\theta_k/\theta_0)U$, an equality in distribution. 
\end{assumption}

Our next assumption for the linear SEM of interest is on linear combinations of the noise terms. This condition is related to Lemma \ref{residsLinearCombo} in the appendix, a key result about how to characterize the regression residuals of Equation \eqref{leastSquaresResids} as linear combinations of ``independent components.''

\begin{assumption}\label{closureConvolution}
For any $j=1,2,\dots,p$ and any $a\in\mathbb{R}^p$ with at least two non-zero entries, 
the linear combination $a^T\epsilon$ does not follow the same distribution as
$\eps_j$.

\end{assumption}

Notable disagreements with Assumption \ref{closureConvolution} are when the $\epsilon_j$ are all Gaussian distributed (not the case for LiNGAM), or when the $\epsilon_j$ are all Poisson-distributed. Notable agreements with Assumption \ref{closureConvolution} (and Assumption \ref{scaleLocation}) are the cases where the $\epsilon_j$ are all Laplace-distributed, all Logistic-distributed, or all Scaled-t distributed (same degrees of freedom). This can be concluded with the characteristic function for a linear combination of two or more $\epsilon_j$'s.

To allow for a quicker sorting procedure in practice, we may make use of an \textit{a priori} known support set for the neighborhood of each node in the DAG.  We consider these neighborhood sets to arise based on domain knowledge, previous studies, or a pre-processing step such as with neighborhood lasso regression of \citet{neighborhoodLasso}. We highlight this usage in Assumption \ref{nbhdsAccurate}: 
\begin{assumption}\label{nbhdsAccurate}
For node $k$, denote its neighborhood estimate as $\MBkhat$. Assume for each $k=1,2,\dots,p$ that:
\[
\MBkhat \supseteq \MBk:= PA_k\cup CH_k\cup \bigcup_{j\in CH_k} PA_j\backslash\{k\},
\]
where $\MBk$ is known as the Markov Blanket of node $k$: the set of its parents $PA_k$, its children $CH_k$, and its co-parents $\bigcup_{j\in CH_k}PA_j\backslash\{k\}$.  \\
\end{assumption}
Let $\MBkthat:=\MBkhat\cap\At,$ which is the subset of the neighborhood set that has been ordered at step $t$ of our procedure (Algorithm~\ref{alg:continueOrdering}). For the cases where $|\MBkthat|\geq 1$, we will make use of least squares residuals for calculating the score $\mathcal{S}(k,\At)$. The corresponding sample version is discussed in \S\ref{FiniteSample}. At the population-level, the residual is
\begin{equation}\label{leastSquaresResids}
R_{kt} :=\begin{cases} X_k&\text{ if }|\MBkthat|=0\\ X_k- \beta^T_{kt} X_{\MBkthat}&\text{ otherwise}\end{cases},
\end{equation}
where $\beta_{kt}$ is the least-squares regression coefficient vector,
\[
\beta_{kt}=\left(\E\left[X_{\MBkthat}X_{\MBkthat}^T\right]\right)^{-1}\E\left[X_{\MBkthat}X_k\right].
\]

\begin{remark}
When we consider the population-level version of our algorithm 
in this section (i.e. we have infinite $n$), we can take $\MBkhat=[p]\backslash\{k\}$ for each $k$ so that Assumption~\ref{nbhdsAccurate} holds trivially. For the finite sample version of our procedure discussed in Section \ref{FiniteSample}, we will make use of Ordinary Least Squares (OLS) linear regressions which require the design matrix to be of full column rank. So if $p\ll n$, we may also take $\MBkhat=[p]\backslash\{k\}$ for each $k$. In the case that $p\gg n$ or $p\approx n$, the neighborhood sets can reduce the number of covariates in OLS regression if $|\MBkhat| \ll n$ for all $k$. 
\end{remark}

\subsection{Our Choice of a Likelihood Ratio Score\label{identifiability}} 
In Algorithm \ref{alg:continueOrdering}, we will select the next node to continue our constructed topological ordering as:
\begin{equation}\label{llr}
\hat{\pi}(t)=\arg\max_{k\not\in\At}\E_{ f_{ kt}(r_{kt}) }\left[\log\frac{ g(\Rkt;\eta_{kt}) }{ \phi(\Rkt;\sigma_{kt} ) }\right].
\end{equation}
Here, $\E_{ f_{kt}(\rkt)}[\cdot]$ denotes expectation with respect to $\Rkt$'s true density, $f_{kt}(\rkt)$. Also,
\[
\eta_{kt}:=\arg\max_\eta\E_{ f_{ kt}(r_{kt}) }\left[\log g(\Rkt;\eta)\right],
\]
while $\phi(r_{kt};\sigma_{kt} )$ is the density for $\mathcal{N}\left(0,\sigma_{kt}^2=\V[\Rkt]\right)$, {i.e. the normal density that matches the mean and variance of $\Rkt$.} Note that $f_{kt}(\rkt)$ is in general different from $g(\rkt;\eta_{kt})$.

The log-likelihood ratio in \eqref{llr} can be thought of as a score that tells us ``how non-Gaussian'' the residual $R_{kt}$ is. If the residual is explained by a Gaussian distribution well relative to the non-Gaussian distribution in the assumed family, then we expect the log-likelihood ratio to be smaller. Otherwise, if the Gaussian density is not a good fit relative to $g(\rkt;\eta_{kt})$, then we have stronger evidence to believe that node $k$ is a valid node to continue the ordering. In Theorem \ref{thm:validityOfLLR}, we claim that using \eqref{llr} leads to the identification of a valid topological ordering, our main result.

\begin{theorem}\label{thm:validityOfLLR} 
Let $X\in\mathbb{R}^p$ follow a LiNGAM with DAG $\mathcal{G}$. If Assumptions \ref{scaleLocation}, \ref{closureConvolution} and \ref{nbhdsAccurate} hold, then applying Algorithm \ref{alg:continueOrdering} at all steps $t=1,2,\dots,p$ with the score 
\[
\mathcal{S}(k,\At) = \E_{f_{ kt}(r_{kt}) }\left[\log\frac{ g(\Rkt;\eta_{kt}) }{ \phi(\rkt;\sigma_{kt}) }\right]
\]
will identify a permutation $\hat{\pi}=(\hat{\pi}(1),\ldots,\hat{\pi}(p))$ that is 
a topological ordering of $\mathcal{G}$. 
\end{theorem}
Theorem \ref{thm:validityOfLLR} suggests that the maximization at each iteration in which we apply Algorithm \ref{alg:continueOrdering} can be done easily. This differs from maximizing a score over a whole ordering which may also lead to identification of the true MEC, but is in general NP hard (not tractable). Relatedly, Appendix \ref{append:scoreBasedGreedy} gives additional motivation for the choice of $\mathcal{S}(k,\At)$ in Theorem \ref{thm:validityOfLLR} as one that allows us greedily optimize the mean log-likelihood when the full ordering is only partially discovered. The proof of Theorem \ref{thm:validityOfLLR} in Appendix \ref{formalProofOfValidity} is an inductive application of key lemmas found in Appendix \ref{identifiabilityAppendix}.

\subsection{Finite Sample Sorting Procedure\label{FiniteSample}}

Assume that we have a data matrix $\X\in\mathbb{R}^{n\times p}$ such that $\X_{i\cdot}$, the $i$-th row, is iid across $i=1,2,\dots,n$ from a distribution defined by a LiNGAM satisfying Assumptions \ref{scaleLocation} and \ref{closureConvolution}. Also let Assumption \ref{nbhdsAccurate} hold, where the sets $\MBkhat$ are given by domain knowledge, or they are estimated with data independent of $\X$ by an asymptotically consistent procedure. 

Denote by $\X_{\cdot S}$ the columns of $\X$ indexed by the set $S$. When $S$ is a singleton, such as $S=\{k\}$, we will simply write $\X_{\cdot k}$ for the $k$-th column. Analogous to Section \ref{identifiability}, consider:
\[
\hat{\beta}_{kt} = \left(\X_{\cdot\MBkthat}^T\X_{\cdot\MBkthat}\right)^{-1}\X_{\cdot\MBkthat}^T\X_{\cdot k}\in\mathbb{R}^{|\MBkthat|\times 1},
\]
which exists so long as $1\leq |\MBkthat|\leq n$ and $\X_{\cdot\MBkthat}$ is of full column rank almost surely. Further, we define $\hat{R}_{kt} \in\mathbb{R}^{n\times 1}$ as 
\[
\hat{R}_{kt}=\begin{cases}\X_{\cdot k}&\text{ if }|\MBkthat|=0\\ \X_{\cdot k}-\X_{\cdot \MBkthat}\hat{\beta}_{kt}&\text{ if }|\MBkthat|\geq 1\end{cases},
\]
the vector of residuals which we will use to estimate the pertinent scale parameter of \eqref{llr}, denoted as $\hat{\eta}_{kt}$ and $\hat{\sigma}_{kt}$, respectively. Explicitly, we select the next node to continue an ordering using the empirical analogue of the mean log-likelihood ratio in Equation \eqref{llr}:
\begin{equation}\label{samplellr}
\hat{\pi}(t)=\arg\max_{k\not\in\At}\frac{1}{n}\sum_{i=1}^n\log\frac{g(\hat{R}_{i,kt};\hat{\eta}_{kt})}{\phi(\hat{R}_{i,kt};\hat{\sigma}_{kt})},
\end{equation}
where $\hat{R}_{i,kt}$ is the $i$-th entry of the vector $\hat{R}_{kt}$, while $\hat{\sigma}^2_{kt}:=\frac{1}{n}\|{\hat{R}_{kt}}\|_2^2$ and $\hat{\eta}_{kt}:=\arg\max_{\eta}\sum_{i=1}^n\log g(\hat{R}_{i,kt};\eta)$. For example, if $\eta_{kt}$ is the scale parameter for a Laplace distribution, it can be seen that $\hat{\eta}_{kt}=\frac{1}{n}\|{\hat{R}_{kt}}\|_1$. In this case, \eqref{samplellr} is equivalent to 
\begin{equation}\label{eq:lapscore}
\hat{\pi}(t)=\arg\max_{k\not\in\At}\log\frac{\hat{\sigma}_{kt}}{\hat{\eta}_{kt}}=\arg\max_{k\not\in\At}\frac{\|{\hat{R}_{kt}}\|_2}{\|{\hat{R}_{kt}}\|_1}.
\end{equation}

{The Laplace update \eqref{eq:lapscore} exemplifies how simple the maximization of our likelihood ratio score is. After the regression of each unsorted node $X_k$, $k\notin\At$, onto $\MBkthat$, we only need to compare the ratio between the two norms of the residual vector $\hat{R}_{kt}$ across unsorted nodes to find $\hat{\pi}(t)$.} Algorithm \ref{algoInPractice} in the Supplementary Material shows the pseudo-code for the sorting procedure we use in practice, with a strategic update of regression residuals using partial regression {that greatly reduces the computation cost.} {We have also provided the details on the estimation of the scale parameters for Logistic and Scaled-t distributions in Appendix \ref{scaleParameterEstimation}.}


\begin{figure*}[ht]
\vskip 0.2in
\begin{center}
\centerline{\includegraphics[width=\textwidth]{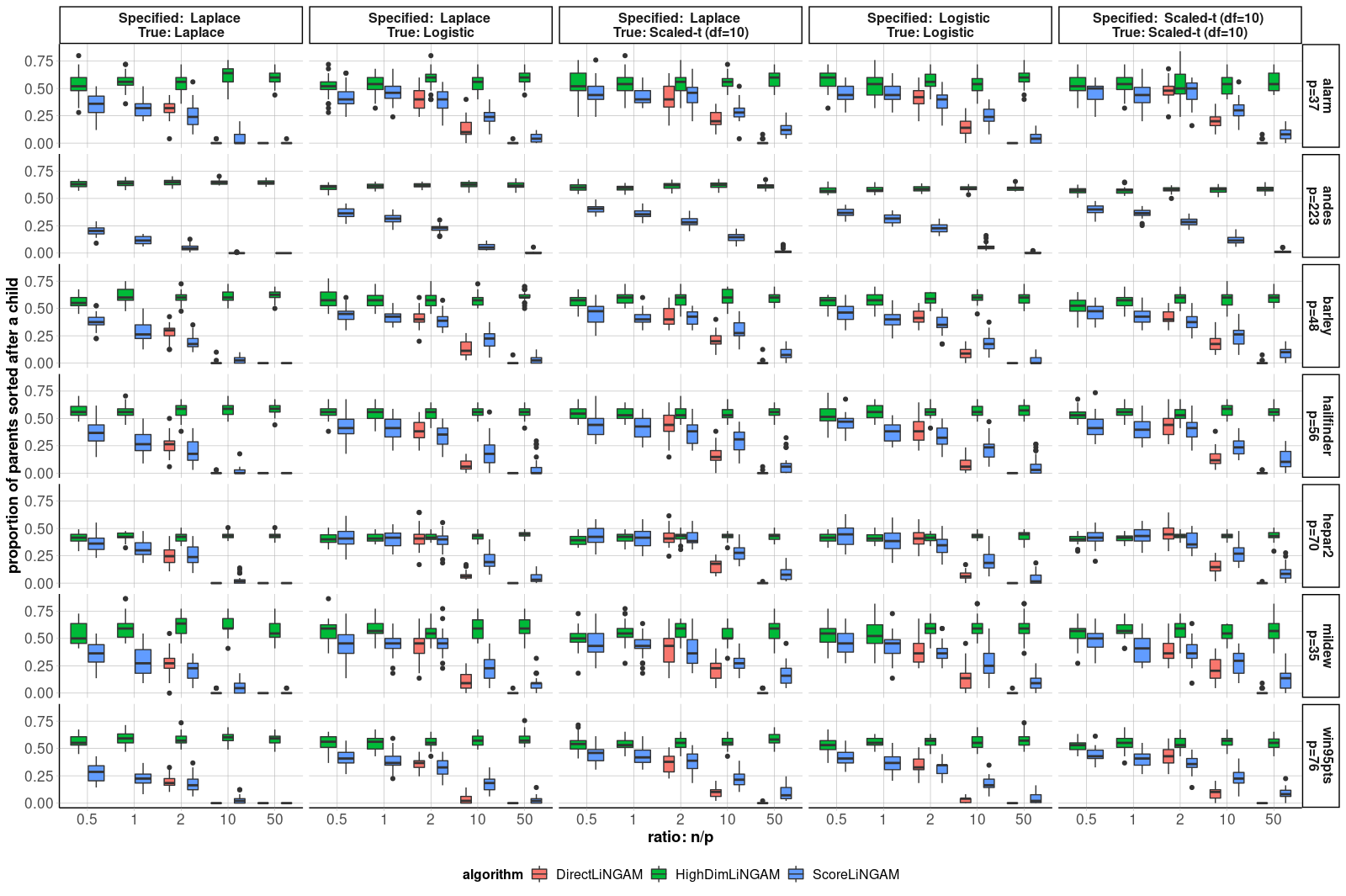}}
\caption{The simulation results comparing LiNGAM estimation procedures.}
\label{fig:sortError}
\end{center}
\vskip -0.2in
\end{figure*}

\section{Empirical Results\label{simulations}}
\subsection{Simulations on Small Networks}
We now present simulation results for networks that are on the smaller end ($35\leq p\leq 223$), {downloaded from the \url{bnlearn.com} Bayesian network repository}. We compared our sorting procedure to other LiNGAM learning procedures. Due to their readily available code, the algorithms of interest are ``DirectLiNGAM'' \citep{shimizu11}, ``HighDimLingam'' \citep{wangHighDLingam}, and ``ScoreLiNGAM'' (our procedure). For each simulation setting, we conduct 30 replicates.

For each choice of $\mathcal{G}$ underlying a LiNGAM, our synthetic data generation schema was as follows. We generated $\mathbf{B}_{jk}\iid\text{Uniform}[-0.9,-0.4]\cup[0.4,0.9]$ for each $(j,k)$ such that $j\in PA_k$, and otherwise set $\mathbf{B}_{jk}=0$. We generated $\theta_k\iid\text{Uniform}[0.4,0.7]$ across $1\leq k\leq p$, {where $\theta_k$ is the scale parameter for the error distributions as in Assumption~\ref{scaleLocation}.} Finally, we varied sample size as $n=0.5p,p,2p,10p,50p$. {Note that $n=0.5p$ and $n=p$ represent the high-dimensional setting ($p\geq n$).} Next, a data set $\X\in\mathbb{R}^{n\times p}$ of iid samples is drawn from the distribution given by the LiNGAM parameterized by $(\mathbf{B},\theta_1,\dots,\theta_p)$ and having errors $\eps_k\sim g(\cdot;\theta_k)$ across $1\leq k\leq p$. Moreover, we varied the family of the densities $g$ in Assumption \ref{scaleLocation} to be the Laplace, the Logistic, or the Scaled-t distribution (10 degrees of freedom) scale-location families.  Finally, ScoreLiNGAM and HighDimLiNGAM were run with knowledge of the true Markov blanket for each node, while DirectLiNGAM was not as it does not have this option. Afterward, the data matrix $\mathbf{X}$ was standardized so that each column has sample standard deviation equal to $1$ and sample mean equal to $0$.

{Figure~\ref{fig:sortError} reports the results in terms of order estimation error (lower is better), which we define as:
\[
\frac{1}{p^2}\sum_{j=1}^p\sum_{ k=1 }^p\mathbf{1}\{ \B_{jk}\neq0 ,\hat{\pi}^{-1}(k)<\hat{\pi}^{-1}(j) \}.
\]
Our ScoreLiNGAM achieved the highest accuracy for all high-dimensional settings ($n\leq p)$. DirectLiNGAM became quite comparable until the sample size increased to $n=2p$ and did a bit better than ScoreLiNGAM when $n\geq 10p$ (large sample size cases).}
Note that results are not presented for DirectLiNGAM when $n=0.5p$ nor $n=p$, because it is not applicable for $n\leq p$. For the Andes network, results for DirectLiNGAM are also not presented as this procedure takes about 118 minutes for a single replicate, which adds up across 90 total replicates. On the other hand, HighDimLiNGAM is generally the least accurate algorithm across all networks and sample sizes. Recall that the data matrix $\mathbf{X}$ is re-scaled. The inaccuracy of HighDimLiNGAM is likely owed to the fact that this procedure is not invariant to a re-scaling of the data, as ScoreLiNGAM and DirectLiNGAM are. {We also compared the three methods when the error distributions were mis-specified for ScoreLiNGAM (second and third columns of Figure~\ref{fig:sortError}). The true error distributions were Logistic or Scaled-t, but we still used the Laplace update \eqref{eq:lapscore} in ScoreLiNGAM. It is seen that its accuracy was comparable to the result when we correctly specified the error distributions (the other three columns), suggesting that our method is robust to model mis-specification.}

In terms of speed, Figure \ref{fig:sortTmWin95Pts} summarizes this for the win95pts network. The advantage of our method is speed, with our method being no less than 100 times faster the next fastest method. (Note: HighDimLiNGAM's procedure is parallelized across 7 threads.) Appendix \ref{moreFigs} contains details about the implementation of each procedure, along with the machine used to run these experiments. Moreover, Figure \ref{fig:sortTm} in Appendix \ref{moreFigs} contains sorting times for all the settings we considered. 


\begin{figure}
\centering
\begin{minipage}{.33\textwidth}
  \centering
  \includegraphics[width=\textwidth]{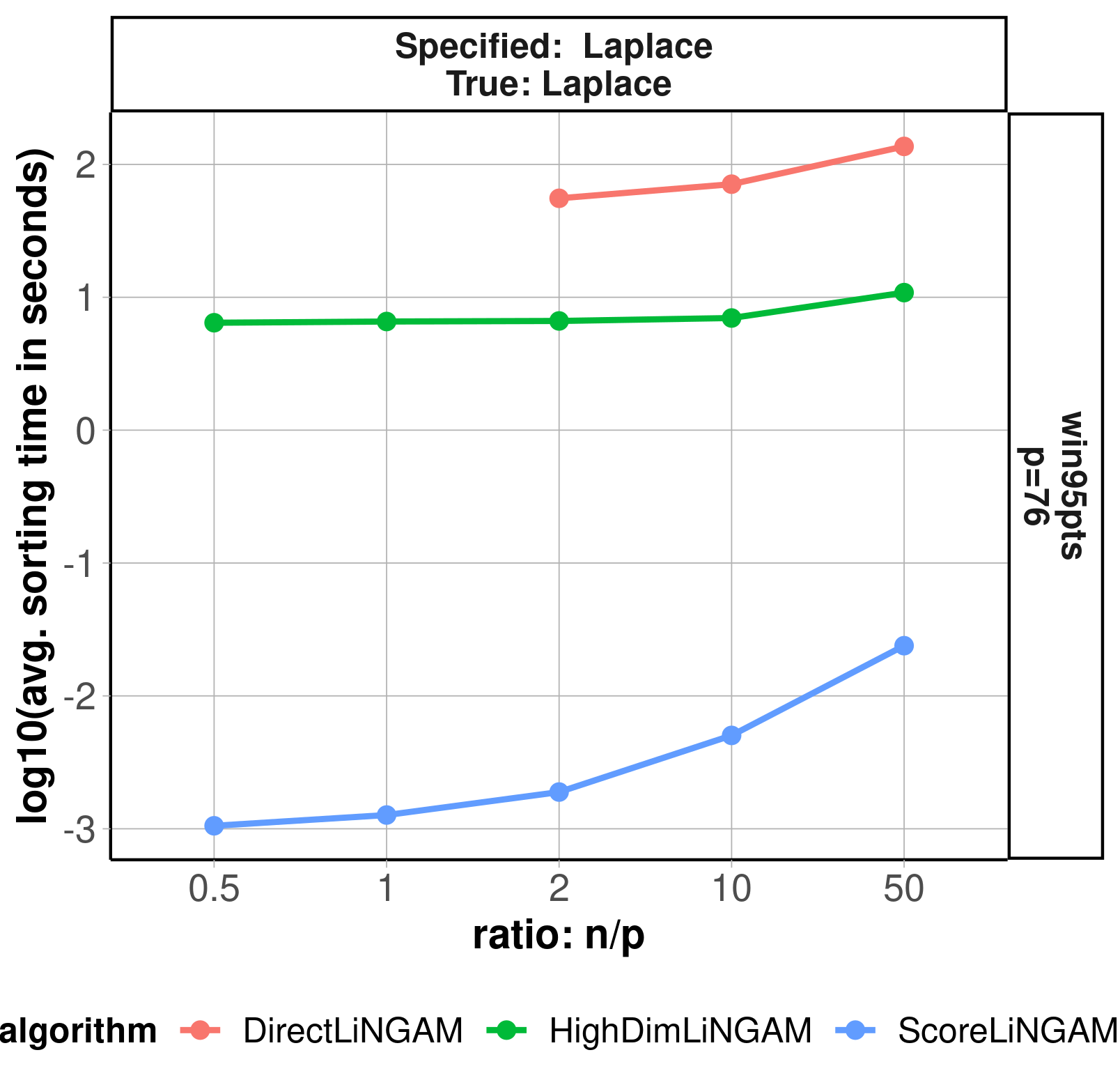}
  \captionof{figure}{The $\log_{10}$(avg. sorting time in seconds) scale for the various methods applied to the win95pts network.}
\label{fig:sortTmWin95Pts}
\end{minipage}%
\hfill
\begin{minipage}{.31\textwidth}
  \centering
  \includegraphics[width=\textwidth]{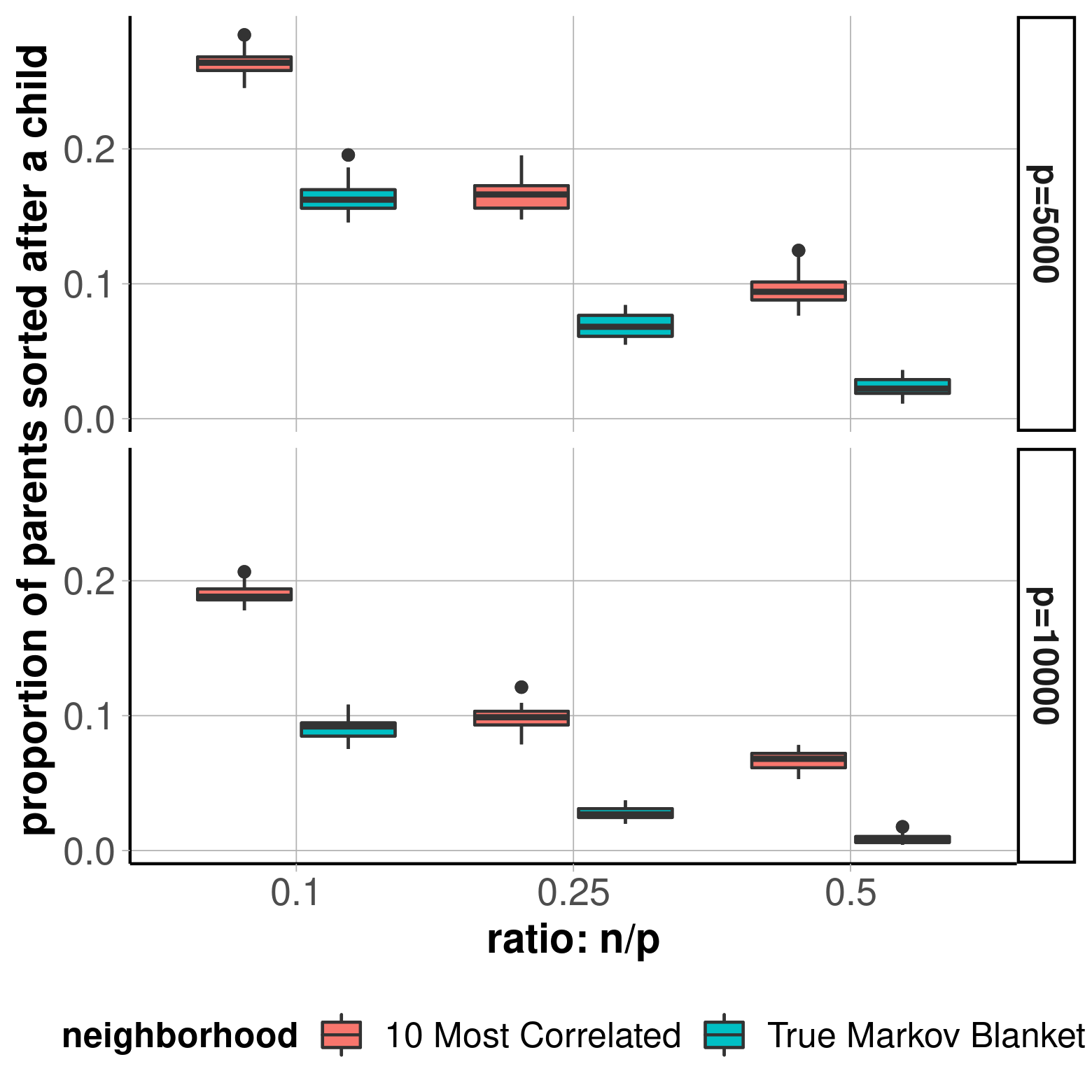}
  \captionof{figure}{Sorting errors for ScoreLiNGAM under $p=5000,10000$ and $n=0.1p,0.25p,0.5p$. Color indicates how the neighborhood sets are constructed.}
\label{fig:highDimResults}
\end{minipage}
\hfill
\begin{minipage}{0.33\textwidth}
\centering
\includegraphics[width=\textwidth]{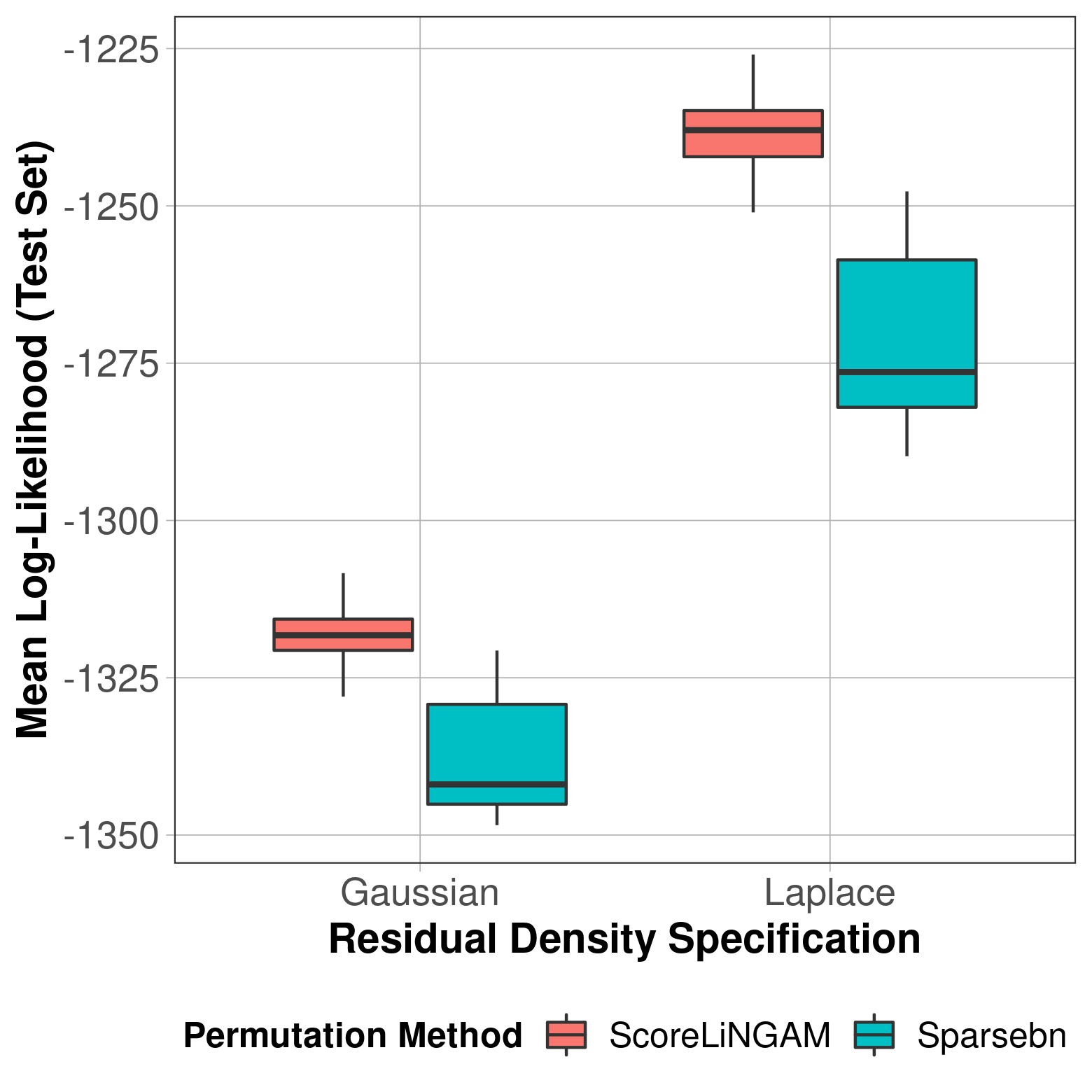}
\caption{The mean log-likelihood on 1,000 genes for a subset of cells in the data of \citep{yao2021scrna-seq}, across 50 repetitions.}
\label{fig:only1000GenesCE}
\end{minipage}
\end{figure}

\subsection{Larger Network Results}\label{sec:largesim}

Next, we simulated large networks with $p=5000,10000$ and $n=0.1p,0.25p,0.5p$ to further demonstrate the scalability of ScoreLiGAM. We do not include results in these settings for DirectLiNGAM nor HighDimLiNGAM as they would take too long to run. The network generation is such that 5\% of nodes are root nodes (no parents), and all other nodes have between 1 and 2 parents (with equi-probability) which are selected at random from the set of predecessors in a randomly generated permutation. Moreover, $\B_{jk}\iid\text{Uniform}[-0.9,-0.4]\cup[0.4,0.9]$ across $(j,k)$ such that $j\in PA_k$, while $\theta_k\iid\text{Uniform}[0.25,0.9]$ across $1\leq k\leq p$ is the scale parameter for the Laplace noise in the synthetic LiNGAM. A new LiNGAM is generated according to this schema for each data replicate. 

Figure \ref{fig:highDimResults} presents simulation results for ScoreLiNGAM 
with two different a priori known neighborhood sets. ``True Markov Blanket'' means that we set $\MBkthat=\MBk$ for each $1\leq k\leq p$ and run the sorting procedure with these oracle sets. The results for ``10 Most Correlated'' use 20\% of the data to specify $\MBkthat$ as the 10 most Pearson-correlated variables (in absolute value) to $X_k$ for each $1\leq k\leq p$, and the other 80\% of the data to estimate the topological ordering. 

It is encouraging to see in Figure~\ref{fig:highDimResults} that the accuracy of our method is high even for such a challenging high-dimensional setting. In fact, the average error rate is quite comparable to that for the smaller networks reported in Figure~\ref{fig:sortError}. 
As expected, an accurate neighborhood set provides better sorting results. Further, our method can run relatively quickly for large $p$, but its accuracy naturally is dictated by sample size. Figure \ref{fig:highDimResultsTm} in Appendix \ref{moreFigs} contains the sorting times to go along with Figure \ref{fig:highDimResults}. 


\subsection{Application: Single-Cell Gene Expression Data\label{applicationSingleCell}}
We apply our method on the data of \citet{yao2021scrna-seq}\footnote{Available at \url{http://cells.ucsc.edu/?ds=allen-celltypes+mouse-cortex&meta=regionlabel} in compressed TSV format}. With it, we seek to estimate a linear SEM to model a gene regulatory network, where each {$X_k$ in Equation \eqref{eq:linSEM}} is the expression level of a gene. We focus our attention on their dataset for which isolated single cells were processed for RNA sequencing using SMART-Seq v4 (labeled ``Mouse Cortex+Hipppocampus (2019/2020)''). Noting the paper's finding that cells' gene expressions cluster according to region and cell type, we subset the data as follows. We focus on glutamatergic cells from the mice brains' primary visual cortex. We also focus on cells for which injection materials are not specified (see \citet{neuronalTracing2019} for background on neuronal tracers). This takes us from 74,973 cells down to 7,159--the largest subset of all cell class, isocortex location, and injection material combinations. A sizable amount of genes had expression measurements of exactly $0$, so we subset genes to those which were measured to be non-zero in 50\% or more of these cells. This brings us from 45,768 to 10,012 genes. 

{As for large simulated networks in Section~\ref{sec:largesim}, DirectLiNGAM and HighDimLiNGAM were too slow for this application.} In order to compare ScoreLiNGAM to another linear structural equation modeling procedure, we applied the package \texttt{sparsebn} \citep{sparsebnJSSv091i11} to our data, which is a score-based method that maximizes a regularized Gaussian likelihood over the DAG space \citep{aragamZhou15a}. To make comparisons across 50 repetitions, we randomly select 1,000 of the original 10,012 genes. For each repetition, we randomly sample 2,000 cells: half of the cells are designated to be in the training set, and the other half in test set; each data matrix is standardized such that columns have sample standard deviation $1$ and sample mean $0$. 

In the training set, 20\% of cells are randomly selected to estimate the Pearson correlation matrix. We specify the neighborhoods, $\MBkhat$, for ScoreLiNGAM as the 50 genes $j\in[1000]\backslash\{k\}$ with the largest Pearson correlation (in absolute value) with gene $k$. The remaining 80\% of training data is used to estimate a topological ordering and the linear SEM's coefficients (via ordinary least squares). For Sparsebn, no a priori neighborhood selection is used: parent sets for the linear SEM are learned with 100\% of the training data using default options in the \texttt{estimate.dag} command, and the selection of the final DAG in the solution path is done by the recommended \texttt{select.parameter} command. For Sparsebn, the linear SEM's model parameters are estimated according to the selected DAG. Moreover, the noise densities we fit to the residuals in the training set are either Gaussian or Laplace. 

As can be seen in Figure \ref{fig:only1000GenesCE}, the Laplace density specification for the additive errors provides a significantly higher mean log-likelihood on the test set compared to a Gaussian density for both methods. This shows that the Laplace distribution, with its thicker tails than the Gaussian distribution, fits this data better. 
Furthermore, ScoreLiNGAM showed substantially higher test-data likelihood than Sparsebn under both error distributions for calculating the likelihood.

\section{Discussion}

In this paper, we demonstrated that sequentially applying Algorithm \ref{alg:continueOrdering} can give promising structure learning results. We demonstrated this with a novel sequential procedure based on parametric specification that provides an alternative to the state of the art for the identifiability and estimation of a linear DAG model with non-Gaussian errors. 
We discussed the conditions, Assumptions \ref{scaleLocation} and \ref{closureConvolution}, under which the proposed causal discovery procedure will identify the valid DAG. We also proposed a relatively simple procedure that can make strategic use of an a priori known neighborhood set for each node. Finally, we presented numerical evidence that our procedure scales to large dimensions, which is otherwise not the case for the state-of-the-art for LiNGAM. We accompanied these simulations with a real-data application. Further extensions of the work presented here include formal statistical guarantees along with extensions of the likelihood ratio approach to nonlinear SEMs. 



%

\newpage

\bibliographystyle{unsrtnat}
\bibliography{references}


\newpage
\appendix

\section{Greedy Choice of a Factor to Optimize the Joint Likelihood Function\label{append:scoreBasedGreedy}}
Let vector $X\sim f(x)$, where $f(x)$ corresponds to the density in Definition \ref{defn:lingam}. Consider $X$'s expected log-likelihood as a function of the permutation $\pi$:
\begin{equation}\label{likelihoodPi}
\mathcal{L}(\pi)=\sum_{j=1}^p\E_{X\sim f(x)}\left[\log g\left(X_{j}-[\B_{\cdot j}^\pi]^TX;\theta^\pi_j\right)\right],
\end{equation}
Here, $\B^{\pi}$ is the acyclic weighted adjacency matrix that arises from a population-level least squares objective such that the $\pi(j)$-th column is given by:
\[
\B^{\pi}_{ \cdot \pi(j) } = \arg\min_{ \substack{\theta\in\mathbb{R}^{p\times 1}:\ \theta_k=0\ \forall k\\\text{ s.t }\pi^{-1}(k)\ \geq\ j  }}\mathbb{E}\left[( X_{\pi(j)}-\theta^TX )^2\right].
\]
That is, the column $\B^{\pi}_{\cdot \pi(j)}$ is comprised of the least squares coefficients when linearly regressing $\pi(j)$ onto its predecessors, if any, in the ordering given by $\pi$. Moreover, $\theta^\pi_j$ is the corresponding scale parameter according to Assumption \ref{scaleLocation}.
Now let $\phi_j^\pi$ be the density for the Gaussian distribution having the same first two moments as:
\[
R_j^{\pi}:=X_j-[\B^{\pi}_{\cdot j}]^T X.
\]
Define
\[
\tilde{\mathcal{L}}(\pi):=\sum_{j=1}^p\E_{X\sim g(x)}\left[\log\phi_{j}^\pi\left(R_j^\pi\right)\right]\text{ and }\kappa:=\E_{\tilde{X}\sim g(x)}\left[\log\mathcal{N}(\tilde{X};\E[X],\V[X])\right].
\]
Here, $\mathcal{N}(x;\E[X],\V[X])$ denotes the density for a $p$-variate Gaussian distribution with the same first and second order moments as $X$. Due to the relation between $\mathbf{B}^\pi$ and the generalized Cholesky factorization of $\V[X]$, \citet{yeSimAnnealing} shows that we actually have the equality:
\begin{equation}\label{eqn:invarianceGaussPi}
\tilde{\mathcal{L}}(\pi)=\kappa.
\end{equation} 

Thus, maximizing \eqref{likelihoodPi} with respect to $\pi$ is the same as maximizing the expected log-likelihood ratio given by: 
\begin{equation}\label{likelihoodRatioPi}
(\mathcal{L}-\tilde{\mathcal{L}})(\pi)=\sum_{j=1}^p\E_{X\sim g(x)}\left[\log\frac{g\left(R_{j}^{{\pi}};\theta^\pi_j\right)}{\phi_{j}^\pi\left(R_j^\pi\right)}\right]=\mathcal{L}(\pi)-\kappa.
\end{equation}

With all this in mind, we can think of our choice of a node to append to the ordering $\At$ at step $t$ as greedily choosing the largest summand,
\[
\E_{X\sim g(x)}\left[\log\frac{g\left(R_{\hat{\pi}(t)}^{\hat{\pi}};\theta^{\hat{\pi}}_{\hat{\pi}(t)}\right)}{\phi_{\hat{\pi}(t)}^{\hat{\pi}}\left(R_{\hat{\pi}(t)}^{\hat{\pi}}\right)}\right]
\]
to add to the known log-likelihood ratio at step $t$:
\[
(\mathcal{L}-\tilde{\mathcal{L}})_t(\hat{\pi}):=\begin{cases}0&t=1\\
\sum_{j=1}^{t-1}\E_{X\sim g(x)}\left[\log\frac{g\left(R_{\hat{\pi}(j)}^\pi;\theta^{\hat{\pi}}_{\hat{\pi}(j)}\right)}{\phi_{\hat{\pi}(j)}^{\hat{\pi}}\left(R_{\hat{\pi}(j)}^\pi\right)}\right]& 2\leq t\leq p+1
\end{cases}.
\]
That is, our sequential application of Algorithm \ref{alg:continueOrdering} is attempting to greedily maximize \eqref{likelihoodRatioPi} one summand at a time.

\section{Proof of Theorem \ref{thm:validityOfLLR}\label{identifiabilityWithLLR}}

\subsubsection{Proof sketch for Theorem \ref{thm:validityOfLLR}}

The formal proof of Theorem \ref{thm:validityOfLLR} in \S~\ref{formalProofOfValidity} below is a relatively straightforward inductive application of the following reasoning after applying Algorithm \ref{alg:continueOrdering} at any given step $t$. Key to the proof, we note that \eqref{llr} can also be written equivalently as the difference of two KL-divergence terms:
\begin{equation}\begin{aligned}\label{diffKL}
\arg\max_{k\not\in\At}& \{D_{KL}\left(f_{kt}(\rkt)\middle|\middle|\phi(\rkt;\sigma_{kt})\right)- D_{KL}\left(f_{kt}(\rkt)\middle|\middle|g_{k}(\rkt;\eta_{kt})\right) \}.
\end{aligned}\end{equation}
Lemma \ref{residsLinearCombo} suggests that invalid nodes' residuals, $\Rkt$, are a linear combination of two or more entries in the vector $\epsilon$, while for valid nodes $\ell$ we have $\Rlt=\epsilon_\ell$. Under Assumption \ref{closureConvolution}, this means that the term $D_{KL}\left(f_{kt}(\rkt)\middle|\middle|g_{k}(\rkt;\eta_{kt})\right)$
in \eqref{diffKL} will be zero only if node $k$ is valid to continue the ordering at step $t$. The natural follow up question is what the behavior is for the term $D_{KL}\left(f_{kt}(\rkt)\middle|\middle|\phi(\rkt;\sigma_{kt})\right)$
in \eqref{diffKL} when $k$ is valid vs. invalid to continue the ordering. Lemma \ref{lemCLTesque} provides this insight: for valid nodes to continue an ordering, this term's value is no less than the same term's value for invalid nodes.

In light of Lemma \ref{residsLinearCombo}, Lemma \ref{lemCLTesque} makes sense under a Central Limit Theorem-like argument: a sum of two or more random variables is closer to Gaussian than each summand alone. Of particular note, a key result that helps show why Lemma \ref{lemCLTesque} holds is \textit{Theorem 17.8.1} of \citep{coverThomas2005}, a restatement of the entropy-power inequality. This restatement says that the differential entropy for a sum of any two independent random variables, $U$ and $V$, is no less than the differential entropy for the sum of two strategically defined Gaussian random variables, each having the same differential entropy as $U$ and $V$ (rather than the same first two moments), respectively.

\subsection{Formal Proof of Theorem \ref{thm:validityOfLLR}\label{formalProofOfValidity}}
\begin{proof}[Proof of Theorem \ref{thm:validityOfLLR}]\ \\
Our proof boils down to making the correct decision in Algorithm \ref{alg:continueOrdering} at step 1, then making the correct choice at step 2 assuming the choice in step 1 was correct, and so on. \\

For the sake of induction, let us assume that $\At$ is correct in the sense that $PA_a\subseteq \At$ for all $a\in\At$. This is true at the base case $t=1$ when $\At=\emptyset$, since having made no ordering choices also means we have made no mistakes.

Let $k\in \St$ be an invalid node to continue the ordering in the sense that $PA_k\cap\At\neq \emptyset$. And let $\ell\in\St$ be a valid node to continue the ordering in the sense that $PA(\ell)\subseteq\At$.

Lemma \ref{lemCLTesque} tells us that the least squares residual $\Rlt\sim f_{\ell t}(\rlt)$ is no closer to Gaussian than $\Rkt\sim f_{kt}(\rkt)$ in the sense that:
\[
D_{KL}\left(f_{kt}(\rkt)\middle|\middle|\phi(\rkt;\sigma_{kt})\right)\leq D_{KL}\left(f_{\ell t}(\rlt)\middle|\middle|\phi_{\ell t}(\rlt)\right)
\]
 
Furthermore, regularity Assumption \ref{closureConvolution} ensures that:
\[
D_{KL}\left(f_{kt}(\rkt)\middle|\middle|g_{k}(\rkt;\eta_{kt})\right)>0.
\]
On the other hand, so long as we properly specified the error density for node $\ell$, we have that:

\[
D_{KL}\left(f_{\ell t}(\rlt)\middle|\middle|g_{\ell}(\rlt;\eta_{\ell t})\right)=0. 
\]

Thus, 

\[\begin{aligned}
&\E_{f_{ kt}(r_{kt}) }\left[\log\frac{ g_k(\Rkt;\eta_{kt}) }{ \phi(\rkt;\sigma_{kt}) }\right]& =\ & D_{KL}\left(f_{kt}(\rkt)\middle|\middle|\phi(\rkt;\sigma_{kt})\right)- D_{KL}\left(f_{kt}(\rkt)\middle|\middle|g_{k}(\rkt;\eta_{kt})\right)\\
<\ & \E_{f_{ \ell t}(\rlt) }\left[\log\frac{ g_\ell(\Rlt;\eta_{\ell t}) }{ \phi_{\ell t}(\Rlt) }\right]&=\ &D_{KL}\left(f_{\ell t}(\rlt)\middle|\middle|\phi_{\ell t}(\rlt)\right).
\end{aligned}\]

Altogether, this implies that

\[
\max_{j\in\St} \mathcal{S}(j,\At) > \mathcal{S}(k,\At).
\]
and
\[
\ell = \arg\max_{j\in\St} \mathcal{S}(j,\At),
\]
since $\ell$ and $k$ were arbitrary valid and invalid nodes, respectively.

So at step $t$, we will always make the correct choice for a node to continue the ordering. 
\end{proof}

\subsection{Proofs of Lemma \ref{residsLinearCombo} and Lemma \ref{lemCLTesque}\label{identifiabilityAppendix}}

In this section, we formally prove Lemma \ref{residsLinearCombo}, Lemma \ref{lemCLTesque}.

\subsubsection{Some Useful Shorthand Notation}

Let us define some new strategic sets which contain indices in $[p]$, and review some we have been using already. 
\begin{itemize}
    \item The set $$\At=\begin{cases}\emptyset&t=1\\
    \{\hat{\pi}(1),\dots,\hat{\pi}(t-1)\}&t\geq 2\end{cases}.$$
    This is the partial ordering at step $t=1,2,\dots$. In our population-level identification results, we will typically assume it is correct at step $t$, which means that for all $a\in\At$, $PA_a\subset\At$.
    \item $\St=[p]\backslash\At$ is the set of unordered nodes at step $t$.
    \item $\MBk = PA_k\cup CH_k \cup_{j\in CH_k } PA_j$ is the Markov Blanket of node $k$.
    \item $\MBkhat$ is the Markov Blanket superset such that $\MBkhat\supseteq\MBk$. In finite data, we will typically estimate $\MBkhat$ by a procedure such as neighborhood lasso regression, so this containment may not hold. For the sake of this section, because we are deriving quantities at the population-level, we assume that $\MBkhat$ is known and contains the true Markov blanket. Note that trivially, we may consider $\MBkhat=[p]\backslash\{k\}$, and the results of this section would still hold.  
    \item $\MBkthat=\At\cap\MBkhat$ is the intersection of the Markov blanket superset with the partial ordering. 
    \item $\Lkt=\bigcup_{j\in\MBkthat}\{j\}\cup AN_j$, which are either nodes of $\MBkthat$ or ancestors of nodes in $\MBkthat$. When $\At$ is correct, it is necessarily the case that $\Lkt\subseteq\At$ for each $k\not\in\At$.
    \item $\Lkt^C$, the complement of set $\Lkt$ which either contains nodes in $\At$ which are not in $\Lkt$, i.e. the nodes of $\At\backslash\Lkt$, or which are unordered, i.e. we have that $\St\subseteq \Lkt^C$.
\end{itemize}

Note that for each node $k\in S_t$ we can write:
\begin{equation}\label{mixturek}
X_k=\M_{k\cdot}\eps = \M_{k \Lkt }\eps_{\Lkt}+\M_{k\Lkt^C}\eps_{ \Lkt^C },
\end{equation}
where the second equality holds since $\Lkt\cup\Lkt^C=[p]$. We can similarly write
\begin{equation}\label{mixtureOrderedNeighbors}
X_{\MBkthat} = \M_{ \MBkthat \cdot}\eps= \M_{ \MBkthat \Lkt}\eps_{ \Lkt }.
\end{equation}
We omit a term with $\epsilon_{\Lkt^C}$ since by definition of $\Lkt$, the sub-mixing matrix $\M_{\MBkthat\Lkt^C}$ is a zero matrix. 

Combining \eqref{mixturek} and \eqref{mixtureOrderedNeighbors},
\[
R_{kt} = \left(\M_{k \Lkt }-\beta^T_{kt} \M_{ \MBkthat \Lkt} \right)\eps_{\Lkt}+\M_{k\Lkt^C}\eps_{ \Lkt^C },
\]
which we will make use of in the proof for Lemma \ref{residsLinearCombo} below. 

\subsubsection{Lemma \ref{residsLinearCombo}: Characterizing nodes' residuals as linear combinations of independent components}
\begin{lemma}\label{residsLinearCombo} Assume that $\At$ is correct so far in the sense that for each $a\in\At$, we have $PA_a\subseteq\At$. Also assume Assumption \ref{nbhdsAccurate} holds. We have that: 
\begin{itemize}
\item If $k\in\St$ is a valid node to continue the ordering, i.e. $PA_k\subseteq \At$, then:
\[
R_{kt} = X_k-\beta_{kt}^TX_{\MBkthat} = \epsilon_k.
\]
\item Otherwise, if $k$ is not a valid node, then $R_{kt}$ is a linear combination of more than one independent component in $\epsilon$.
\end{itemize}

\begin{proof}[Proof of Lemma \ref{residsLinearCombo}]\ \\
\textbf{Case 1:} Assume $k$ is a valid node to continue the ordering in the sense that $PA_k\subseteq \At$. We want to show that $R_{kt}=\epsilon_k$. In this case, $\M_{k\Lkt^C}$ has a non-zero entry corresponding to only $\M_{kk}=1$. This is because $AN_k =\text{support}\left(\textbf{M}_{k\cdot}\right)\backslash\{k\} \subseteq \Lkt$, which in turn holds because $PA_k\subseteq MB_k\cap\At \subseteq \MBkthat$. Thus we have
\[
\M_{k\Lkt^C}\eps_{ \Lkt^C }=\epsilon_k.
\]
So we have left to show that 
\[
\left(\M_{k \Lkt }-\beta^T_{kt} \M_{ \MBkthat \Lkt} \right)\eps_{\Lkt}=0.
\]

Recall that $\B$ is the weighted adjacency matrix for the underlying LiNGAM. We have that $\text{support}(\B_{\cdot k})= PA_k$. Let us index the entries of the column vector $\B_{\cdot k}$ by $\MBkthat$ and denote this as $\B_{\MBkthat k }$. One thing that could be helpful to prove is that if $k$ is valid, then:
\[
\beta_{kt} = \B_{\MBkthat k }. 
\]

Because $\text{support}(\B_{\cdot k})=PA_k$ and $PA_k\subseteq \MBkthat$, consider that
\[
X_k = X^T\B_{\cdot k} + \epsilon_k = X_{\MBkthat}^T\B_{\MBkthat k} + \epsilon_k,
\]
with $\eps_k\indep X_{ \MBkthat k }$ and $\E[\eps_k]=0$. Thus,
\begin{equation}
\begin{aligned}
&\beta_{kt}&=\ &\left(\E\left[X_{\MBkthat}X_{\MBkthat}^T\right]\right)^{-1}\left(\E\left[X_{\MBkthat}X_{\MBkthat}^T\right]\B_{\MBkthat k}+\E\left[X_{\MBkthat}\eps_k\right]\right)\\
&&=\ &\left(\E\left[X_{\MBkthat}X_{\MBkthat}^T\right]\right)^{-1}\E\left[X_{\MBkthat}X_{\MBkthat}^T \right] \B_{\MBkthat k}\\
&&=\ & \B_{\MBkthat k},
\end{aligned}
\end{equation}
as we wanted. 


{It follows that $X_k=\mathbf{B}_{\MBkthat k }^TX_{\MBkthat}+\eps_k=\beta_{kt}^TX_{\MBkthat}+\eps_k$. This then means that $R_{kt}=X_k-\beta_{kt}^TX_{\MBkthat}=\eps_k$, as we wanted to show. }


\textbf{Case 2:} Assume $k$ is not a valid node. All we need in this case for our identifiability proof is that $R_{kt}$ is a linear combination of more than one independent component. This is the case because if $k$ is invalid to continue the ordering, then we have that there exists at least one $j\in PA_k$ such that $j\in S_t$ (unordered) and therefore $j\in \Lkt^C$. Recall that:

\[
R_{kt} = \left(\M_{k \Lkt }-\beta^T_{kt} \M_{ \MBkthat \Lkt} \right)\eps_{\Lkt}+\M_{k\Lkt^C}\eps_{ \Lkt^C }.
\]

Note that it is necessarily the case that $\M_{kj}\neq 0$, otherwise $j\not\in PA_k$. Thus, $R_{kt}$ includes the sum $\M_{kj}\eps_j+\eps_k$. That is, $R_{kt}$ in this case is a linear combination of more than one independent component in $\eps$. Note that $R_{kt}$ could be a linear combination of more entries in $\eps$, in addition to $\eps_j$ and $\eps_k$. \\



\end{proof}
\end{lemma}

\subsubsection{Some Information Theory Definitions and Results}

We now present some straightforward information theoretic results. They are meant to help demonstrate that our surrogate optimization (now a likelihood ratio) approach for Algorithm \ref{alg:continueOrdering} leads to the identifiability of a causal order. These lemmas are used later to prove Lemma \ref{lemCLTesque}, a key result that says valid nodes $j$ in a LiNGAM are no closer to Gaussian compared to invalid nodes $k$, conditional on the nodes in $\widehat{N}_{jt}$ and $\MBkthat$, respectively.

\begin{definition}[Differential Entropy]\ \\
For a continuous random variable $X$ with density ${p}(x)$, denote $\E_{p(x)}\left[\cdot\right]$ to be expectation with respect to $p(x)$. The differential entropy of $X$ is given by: 
\[
\mathbf{h}\left(X\right)=\mathbb{E}_{{p}(x)}\left[\log\frac{1}{{p}(X)}\right].
\]

\end{definition}

\begin{lemma}[\label{entPow}Restatement of the entropy power inequality]\ \\
Consider two independent random variables $X\sim{p}(x)$ and $Y\sim{p}(y)$, and let $X'\sim\mathcal{N}\left(\mathbb{E}[X'],\V[X']\right)$ and $Y'\sim\mathcal{N}\left(\mathbb{E}[Y'],\V[Y']\right)$ be independent random variables such that $\mathbf{h}(X)=\mathbf{h}(X')$ and $\mathbf{h}(Y)=\mathbf{h}(Y')$. Then:
\[
\mathbf{h}(X+Y)\geq\mathbf{h}(X'+Y').
\]
\begin{proof}
This is exactly \textit{Theorem 17.8.1} of \citep{coverThomas2005}, so we refer the reader to their proof.\\
\end{proof}
\end{lemma}

\begin{lemma}[\label{kl2gauss}KL Divergence from Gaussianity]\ \\
Let $X\sim {p}(x)$ and $q(x)$ the density for $\tilde{X}\sim\mathcal{N}\left(\mathbb{E}\left[X\right],\text{Cov}\left[X\right]\right)$. 
\begin{equation}
D_{KL}\left(p(x)\middle|\middle|q(x)\right)=\mathbf{h}(\tilde{X})-\mathbf{h}\left(X\right).
\end{equation}

As in Lemma \ref{entPow}, let $X'\sim\mathcal{N}\left(\mathbb{E}[X'],\V[X']\right)$ such that $\mathbf{h}(X)=\mathbf{h}(X')$. We can equivalently write the KL divergence from Gaussianity as:
\[
D_{KL}\left(p(x)\middle|\middle|q(x)\right)=\frac{1}{2}\log\left( \frac{ \V[X]}{ \V[X']} \right).
\]
\begin{proof}\ \\
Because
\[
\mathbf{h}(\tilde{X})=\mathbb{E}_{\tilde{X}\sim q(x)}\left\{\log\frac{1}{q(\tilde{X})}\right\}=\mathbb{E}_{X\sim p(x)}\left\{\log\frac{1}{q(X)}\right\},
\]
by properties of this normal distribution (namely, that $\E[\log q(X)]\propto \V[X]=\V[\tilde{X}]$) we have that:
\[
D_{KL}\left(p(x)\middle|\middle|q(x)\right)=\mathbf{h}(\tilde{X})-\mathbf{h}\left(X\right).
\]

Noting that the differential entropy for any $\mathcal{N}(\mu,\sigma^2)$ is $\frac{1}{2}\log(2\pi e \sigma^2)$ and our assumption that $\mathbf{h}(X)=\mathbf{h}(X')$, we arrive at the second equality:
\[
D_{KL}\left(p(x)\middle|\middle|q(x)\right)=\frac{1}{2}\log\left( \frac{2\pi e \V[X]}{2\pi e \V[X']} \right).
\]
Note that also $\V[X]=\V[\tilde{X}]\geq \V[X'] \iff D_{KL}\left(p(x)\middle|\middle|q(x)\right)\geq 0,$ which is the case because KL-divergence is always non-negative. 

\end{proof}
\end{lemma}

This well known result also implies that the normal distribution is the maximum entropy distribution when we constrain the first and second order moments of each distribution to be the same.

\begin{lemma}[Same distance to Gaussianity\label{dist2gauss}]\ \\
Let $\tilde{\eps}_k\sim\mathcal{N}(0,\V[\eps_k])$ with density $q_k(\cdot)$ for each $k=1,2,\dots,p$. Also let $\eps_k'\sim\mathcal{N}(0,\V[\eps'_k])$ such that $\mathbf{h}(\eps_k')=\mathbf{h}(\eps_k)$. If $\eps$ in our LiNGAM satisfies Assumption \ref{scaleLocation}, then there exists a constant $\gamma\geq 0$ such that
\[
D_{KL}\left(g(\eps_k;\theta_k)\middle|\middle|q_k(\eps_k)\right)=\gamma
\]
and
\[
\frac{\V[\eps_k]}{\V[\eps_k']}=\tilde{\gamma}=\exp(2\gamma)
\]
for all $k=1,2,\dots,p$.
\begin{proof}\ \\
From Lemma \ref{kl2gauss}, we have that:
\[
D_{KL}\left(g(\eps_k;\theta_k)\middle|\middle|q_k(\eps_k)\right)=\mathbf{h}(\tilde{\eps}_k)-\mathbf{h}(\eps_k').
\]
Noting Assumption \ref{scaleLocation} and properties of differential entropy under a rescaling, it follows that for $U\sim g(\cdot;\theta_0)$:
\[
\mathbf{h}(\eps_k)=\mathbf{h}(U)+\log(\theta_k/\theta_0).
\]
Let $U'\sim\mathcal{N}(0,\V[U'])$ such that $\mathbf{h}(U')=\mathbf{h}(U)$. We have also that
\[
\mathbf{h}(\eps_k')=\mathbf{h}(U')+\log(\theta_k/\theta_0),
\]
based on the construction of both $\eps_k'$ and $U'$.

Similar to $\tilde{\eps}_k$, let $\tilde{U}\sim\mathcal{N}(0,\V[U])$. Thus, regardless of $k=1,2,\dots,p$, we have that:
\[
\frac{1}{2}\log\left(\frac{2\pi e\V[\eps_k]}{2\pi e\V[\eps_k']}\right)=\mathbf{h}(\tilde{\eps}_k)-\mathbf{h}(\eps_k')=\mathbf{h}(\tilde{U})-\mathbf{h}(U')=:\gamma.
\]

\end{proof}

\end{lemma}

\subsubsection{Lemma \ref{lemCLTesque}: KL Divergence from Gaussianity for valid and invalid nodes' residuals}

\begin{lemma}\label{lemCLTesque} Let $X\in\mathbb{R}^p$  be a LiNGAM from Definition \ref{defn:lingam} that satisfies Assumptions \ref{nbhdsAccurate} and \ref{scaleLocation}. Assume that $\At$ is correct in the sense that $PA_a\subseteq \At$ for all $a\in\At$. Let $k\in [p]\backslash\At$ be an invalid node to continue the ordering in the sense that there exists $j\in PA_k$ such that $j\in[p]\backslash\At$. And let $\ell\in[p]\backslash\At$ be a valid node to continue the ordering in the sense that $PA(\ell)\subseteq\At$. Then the least squares residual $\Rlt\sim f_{\ell t}(\rlt)$ is no closer to Gaussian than $\Rkt\sim f_{kt}(\rkt)$ in the sense that:
\begin{equation}
D_{KL}\left(f_{kt}(\rkt)\middle|\middle|\phi(\rkt;\sigma_{kt})\right)\leq D_{KL}\left(f_{\ell t}(\rlt)\middle|\middle|\phi_{\ell t}(\rlt)\right),
\end{equation}
where $\phi_{k t}$ and $\phi_{\ell t}$ are the respective densities for
\[
\tilde{R}_{kt}\sim\mathcal{N}(\mathbb{E}[\Rkt],\V[\Rkt])\text{ and }\tilde{R}_{\ell t}\sim\mathcal{N}(\mathbb{E}[\Rlt],\V[\Rlt]).
\]
\begin{proof}[Proof of Lemma \ref{lemCLTesque}]\ \\
For each $j\in[p]$, let $\epsilon'_j$ be a normally distributed random variable such that $\mathbf{h}(\epsilon_j')=\mathbf{h}(\epsilon_j)$, while $\tilde{\eps}_j$ is distributed as $\mathcal{N}(\E[\eps_j],\V[\eps_j])$. Here, for all $j,k\in\{1,2,\dots,p\}$, $\tilde{\eps}_j\indep\tilde{\eps}_k$ (unless $j=k$) and $\tilde{\eps}_j\indep\eps_k'$ (even if $j=k$).

Recall also that for $j\in\St$
\[
\Rjt = \left(\M_{j \Lkt }-\beta^T_{jt} \M_{ \MBjthat \Ltj} \right)\eps_{\Ltj}+\M_{j\Ltj^C}\eps_{ \Ltj^C } = \sum_{i\in[p]}\delta_{ij}\eps_i,
\]

where the coefficients $\delta_{ij}$ in the last equality are used for shorthand. Note that $\delta_{jj}=1$ always. And if $j$ is invalid to continue the ordering, then also $\delta_{ij}\neq 0$ for at least one other $i\in[p]\backslash\{j\}$, based on Lemma \ref{residsLinearCombo}.

The relation between the quantities of interest is as follows:
\begin{equation}\footnotesize
\begin{aligned}
&D_{KL}\left(f_{kt}(\rkt)\middle|\middle|\phi(\rkt;\sigma_{kt})\right) & =\ &\mathbf{h}(\tilde{R}_{kt} ) -\mathbf{h}\left(R_{kt}\right) &\text{ by Lemma \ref{kl2gauss} }\\
& & =\ &\mathbf{h}\left(\sum_{i\in[p]}\delta_{ik}\tilde{\eps}_i \right) -\mathbf{h}\left(\sum_{i\in[p]}\delta_{ik}{\eps}_i\right) &\text{ Notice: }\tilde{R}_{kt}\overset{d}{=}\sum_{i\in[p]}\delta_{ik}\tilde{\eps}_i\\
&&\leq\ &\mathbf{h}\left(\sum_{i\in[p]}\delta_{ik}\tilde{\eps}_i \right) -\mathbf{h}\left(\sum_{i\in[p]}\delta_{ik}\eps'_i\right) &\text{ by Lemma \ref{entPow} }\\
&&=\ & \frac{1}{2}\log\left( \frac{2\pi e \sum_{i\in[p] }\delta_{ik}^2\var(\eps_i) }{2\pi e \sum_{i\in [p] }\delta_{ik}^2\var(\eps'_i)} \right)&\text{ by normality of the }\tilde{\eps}_i,\eps_i'\\
&&=\ & \frac{1}{2}\log\left( \frac{2\pi e \tilde{\gamma}\sum_{i\in [p] }\delta_{ik}^2\var(\epsilon'_i) }{2\pi e \sum_{i\in[p] }\delta_{ik}^2\var(\epsilon'_i)} \right)&\text{ by Lemma \ref{dist2gauss} }\\
&&=\ & \gamma &\\
&&=\ &D_{KL}\left(f_{\ell t}(\rlt)\middle|\middle|\phi_{\ell t}(\rlt)\right),&\\
\end{aligned}
\end{equation}
as we wanted (Recall $R_{\ell t}=\epsilon_\ell$ by Lemma \ref{residsLinearCombo}). 

\end{proof}

\end{lemma}

\section{More Figures\label{moreFigs}}

\subsection{Sorting Time for Small Networks}
The general takeaway of Figure \ref{fig:sortTm} is that ScoreLiNGAM is generally much faster. Consider the largest DAG, the Andes network ($p=223$), where the sorting time of ScoreLiNGAM is typically under 1 second across all sample sizes, while for HighDimLiNGAM (parallelized across 7 threads) the sorting procedure takes between 10-1000 seconds across sample sizes. We note that ScoreLiNGAM is written with \texttt{C++} using the \texttt{Armadillo} linear algebra library and an \texttt{R} wrapper via the \texttt{Rcpp} package, while DirectLiNGAM is written in \texttt{Python} (\url{https://github.com/cdt15/lingam}) with a wrapper function in \texttt{R} using the \texttt{reticulate} package that is written by this paper's authors. HighDimLiNGAM is also written in \texttt{C++} (\url{https://github.com/ysamwang/highDNG}) with an \texttt{R} wrapper, but it searches regressor subsets when computing low-dimensional linear regressions--the likely reason for its slower time despite 7 parallel threads.  All simulations were run on a Dell XPS 13 with Intel Core™ i7-8550U CPU @ 1.80GHz × 8, 8 GB RAM, and 64-bit Ubuntu 20.04.3 LTS OS.

\begin{figure*}[ht]
\vskip 0.2in
\begin{center}
\centerline{\includegraphics[width=\textwidth]{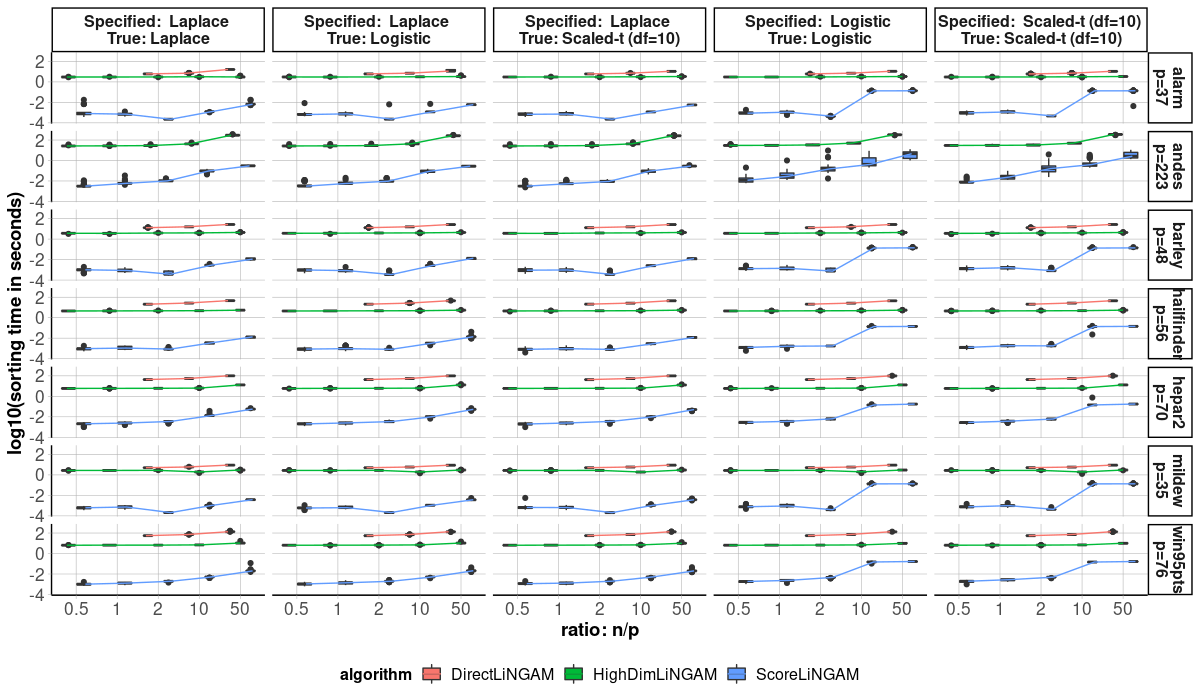}}
\caption{The simulation times for LiNGAM estimation procedures\label{fig:sortTm}}
\end{center}
\vskip -0.2in
\end{figure*}

\subsection{Sorting Times for Large Networks}

Figure \ref{fig:highDimResultsTm} contains the sorting times to go along with Figure \ref{fig:highDimResults} in the main text. 
\begin{figure}[ht]
\vskip 0.2in
\begin{center}
\centerline{\includegraphics[width=0.5\textwidth]{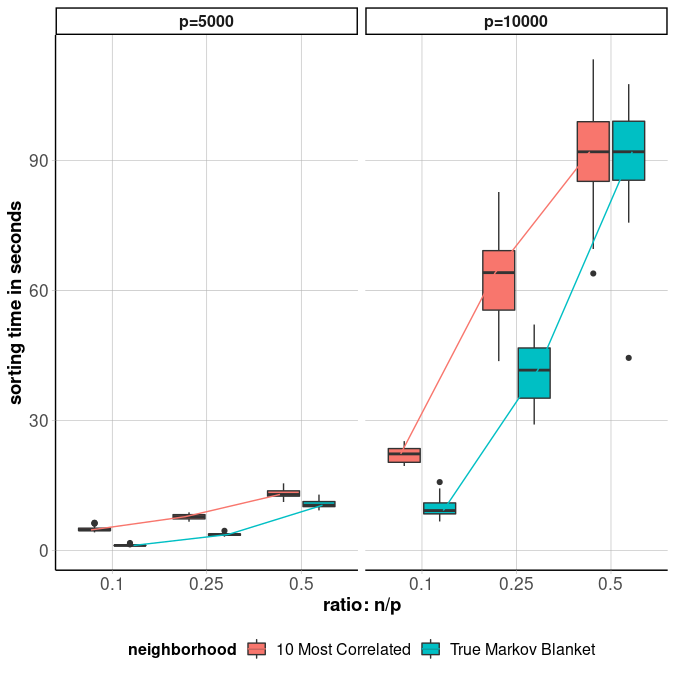}}
\caption{Sorting times for ScoreLiNGAM under $p=5000,10000$ and $n=0.1p,0.25p,0.5p$. Color indicates how the neighborhood sets are constructed.}
\label{fig:highDimResultsTm}
\end{center}
\vskip -0.2in
\end{figure}

\section{The sorting algorithm in practice}

In Algorithm \ref{algoInPractice}, we present further pseudo-code for ScoreLiNGAM's sorting procedure in practice, which uses partial regression.

\begin{algorithm}\ \\
\caption{The sorting procedure in practice}\label{algoInPractice}
\KwData{$\X\in\mathbb{R}^{n\times p}$ (standardized), $\{\MBkhat\}_{k=1}^p$}
\KwResult{$\hat{\pi}(1),\hat{\pi}(2),\dots,\hat{\pi}(p)$}
\# initialize mixing matrix\\
$\mathbf{M}\gets \mathbb{I}_{p \times p}$\\
\# initialize residual matrix\\
$\mathbf{R}\gets \X$\\
\# initialize scores\\
$s_k\gets \mathcal{S}(k;\mathbf{R}),k=1,2,\dots,p$.\\
\# sort the nodes\\
\For{$t=1,2,\dots,p+1$}{
    $\hat{\pi}(t)\gets \arg\max_{k\not\in\At}s_k$\\
    \# update residuals for neighbors of selected node.\\
    \For{$k \in \widehat{N}_{ \hat{\pi}(t) }\backslash\At$}{
        \# update residuals with partial regression.\\
        \For{$a\in \{j:\ \mathbf{M}_{\hat{\pi}(t)j}\neq 0,\mathbf{M}_{kj}=0\}$}{
            $\mathbf{M}_{ka}\gets (\mathbf{R}_{\cdot a}^T\mathbf{R}_{\cdot a})^{-1}\mathbf{R}_{\cdot a}^T\mathbf{R}_{\cdot k}$\\
            $\mathbf{R}_{\cdot k}\gets \mathbf{R}_{\cdot k}-\mathbf{M}_{ka}\mathbf{R}_{\cdot a}$\\
        }
        \# update score\\
        $s_k\gets \mathcal{S}(k;\mathbf{R})$\\
    }
}
\end{algorithm}

\subsection{Obtaining the scale-parameter for Emprical Mean Log-likelihood in \eqref{samplellr}\label{scaleParameterEstimation}}
As discussed in the main text, our sequential algorithm at step $t\geq1$ in practice requires the estimation of the scale parameter, $\eta_{kt}$, in Equation \eqref{samplellr}. Here, we discuss the estimator for the three parametric assumptions used in this paper. We make use of the respective definitions and properties in \citet{Forbes2010-hp}.

\begin{itemize}
    \item \textbf{Laplace Distribution:} If $\eps_k\sim\text{Laplace}(0,\theta_k)$, we have that $\theta_k=\E[|\eps_k|]$ is the scale parameter. When $g(\cdot;\eta_{kt})$ is specified as the density for $\text{Laplace}(0,\eta_{kt})$, the maximum likelihood estimator we use in practice is $\hat{\eta}_{kt}=\frac{1}{n}\norm{\hat{R}_{kt}}_1$.
    \item \textbf{Logistic Distribution:} If $\eps_k\sim\text{Logistic}(0,\theta_k)$, then $\theta_k$ is the scale parameter. We have that $\var[\eps_k]=\frac{\pi^2}{3}\theta_k^2$. When $g(\cdot;\eta_{kt})$ is specified as the density for $\text{Logistic}(0,\eta_{kt})$, we find that the plug-in estimator $\hat{\eta}_{kt}=\frac{\sqrt{3}}{\pi}\hat{\sigma}_{kt}$ to work satisfactorily.
    \item \textbf{Scaled-t Distribution:} If $\eps_k\sim\text{Scaled-t}(0,\nu,\theta_k)$, then we say $\eps_k$ is equal in distribution to the scale parameter, $\theta_k$, times $U\sim \text{t}(0,\nu)$, a Student's t-distributed random variable having mean $0$ and degrees of freedom $\nu>0$. That is, $\eps_k\overset{d}{=}\theta_k U$. For $\nu>2$, we have that $\var[\eps_k]=\theta_k^2\left(\frac{\nu}{\nu-2}\right)$. When $g(\cdot;\eta_{kt})$ is specified as the density for $\text{Scaled-t}(0,\nu,\eta_{kt})$ with $\nu>2$ assumed to be known, we find that the plug-in estimator $\hat{\eta}_{kt}=\hat{\sigma}_{kt}\sqrt{\frac{\nu-2}{\nu}}$ to work satisfactorily.
\end{itemize}
In equation \eqref{samplellr} and in the plug-in estimators for the Logistic and Scaled-t specifications, we use 
\[
\hat{\sigma}_{kt}^2=\frac{1}{n}\norm{\hat{R}_{kt}}_2^2.
\]

\

\end{document}